\documentclass[man,floatsintext,nolmodern,a4paper]{apa7}

\usepackage{amsmath,amssymb,amsfonts,amsthm}
\usepackage{mathtools}
\usepackage{natbib}
\usepackage{url}
\usepackage{enumitem}
\usepackage{booktabs}
\usepackage{xcolor}
\usepackage{tikz}
\usepackage{pgfplots}
\usepackage{float}
\usepackage{mdframed}
\usepackage{tabularx}
\usepackage{ragged2e}

\pgfplotsset{compat=1.18}
\usetikzlibrary{arrows.meta, positioning, calc, patterns}

\theoremstyle{plain}
\newtheorem{theorem}{Theorem}
\newtheorem{lemma}[theorem]{Lemma}
\newtheorem{proposition}[theorem]{Proposition}
\newtheorem{corollary}[theorem]{Corollary}

\theoremstyle{definition}
\newtheorem{definition}[theorem]{Definition}
\newtheorem{example}[theorem]{Example}

\theoremstyle{remark}
\newtheorem{remark}[theorem]{Remark}

\newcommand{\E}{\mathbb{E}}
\newcommand{\Prob}{\mathbb{P}}
\DeclareMathOperator{\PPV}{PPV}
\DeclareMathOperator{\NPV}{NPV}

\newcommand{\Lreq}{\Lambda_{\mathrm{req}}}
\newcommand{\Leff}{\Lambda_{\mathrm{eff}}}
\newcommand{\pcrit}{\pi_{\mathrm{crit}}}
\newcommand{\aeff}{\alpha_{\mathrm{eff}}}

\title{The Certainty Bound: Structural Limits on Scientific Reliability}
\shorttitle{The Certainty Bound}

\author{Marco Pollanen}
\authorsaffiliations{Department of Mathematics and Statistics, Trent University, Peterborough, ON K9L 0G2, Canada}

\authornote{
  ORCID: Marco Pollanen, \url{https://orcid.org/0000-0001-5356-1889}

  Correspondence: marcopollanen@trentu.ca

  Submitted to \textit{Meta-Psychology}. Participate in open peer review by sending an email to open.peer.reviewer@gmail.com. The full editorial process of all articles under review at \textit{Meta-Psychology} can be found following this link: \url{https://tinyurl.com/mp-submissions}. You will find this preprint by searching for the author's name.
}

\begin{document}

\maketitle

\begin{abstract}
\noindent Explanations of the replication crisis often emphasize misconduct, questionable research practices, or incentive misalignment, implying that behavioral reform is sufficient. This paper argues that a substantial component of the crisis is architectural: within binary significance-based publication architectures, even perfectly diligent researchers face structural limits on the reliability they can deliver.

The posterior log-odds of a finding equals prior log-odds plus $\log\Lambda$, where $\Lambda = (1-\beta)/\alpha$ is the experimental leverage. Interpreted architecturally, this implies a hard constraint: once evidence is coarsened to a binary significance decision, the decision rule contributes exactly $\log\Lambda$ to posterior log-odds. A target reliability $\tau$ is feasible if and only if $\pi \geq \pcrit$, and under fixed $\alpha$ this condition cannot, in general, be rescued by sample size alone. Two distinct mechanisms can drive effective leverage to~1: persistent unmeasured confounding in observational studies and unbounded specification search under publication pressure, without requiring bad faith. These results concern binary significance-based decision architectures and do not bound inference based on full likelihoods or richer continuous evidence summaries. Two collapse results formalize these mechanisms, while the Replication Pipeline Theorem and Minimum Pipeline Depth Corollary identify a quantitative evidentiary standard for escape.

Applied to independently documented parameters for pre-reform psychology ($\pi \approx 0.10$, power $\approx 0.35$), the framework implies a replication rate of 36\%, consistent with the Open Science Collaboration's figure. The framework also yields quantitative bridges to the philosophy of science, including Popperian falsification, Kuhnian paradigm shifts, and Lakatosian degenerative programmes. In low-prior settings below the single-study feasibility threshold, the natural unit of evidence is the replication pipeline rather than the individual experiment.
\end{abstract}

\bigskip


\newpage

%==========================================================
% SECTION 1: INTRODUCTION
%==========================================================
\section{Introduction}
\label{sec:intro}

\subsection{A Mystery and Its Resolution}

In 2015, the Open Science Collaboration reported the results of 100 replication attempts in psychology. Under the significance-based replication criterion, 36\% of replication attempts produced a significant result in the same direction \citep{osc2015}. Much of the early reform discourse framed this outcome behaviorally: researchers were said to p-hack, selectively report, and respond to incentives that reward novelty over rigor. The implied remedy was therefore also behavioral: better scientists, better journals, and better norms.

This paper offers a different account, one that is more uncomfortable but also more tractable. The 36\% figure need not be read primarily as a symptom of bad behavior. It is consistent with a structural outcome implied by a mathematical identity under independently documented operating conditions. On this view, the central problem lies not in the moral character of individual researchers, but in the configuration of the evidential architecture.

The reliability of a statistically significant finding is governed by an accounting identity. Let $\pi$ denote the prior probability that a tested hypothesis is true and $\Lambda = (1-\beta)/\alpha$ the \emph{experimental leverage}, the ratio of power to the false-positive rate. Then the positive predictive value satisfies $\PPV = \pi\Lambda/[\pi\Lambda + (1-\pi)]$, and a test contributes exactly $\log\Lambda$ to the posterior log-odds of a genuine effect. This is Bayes' theorem in a familiar form. The contribution here is architectural rather than algebraic: once evidence is coarsened to a binary significance decision, the experiment contributes exactly $\log\Lambda$ to posterior log-odds. Within that architecture, the same decision rule cannot extract additional evidential gain. If $\pi$ is small or $\Lambda$ is modest, the posterior probability remains low regardless of sample size or investigator diligence.

Parameter ranges documented in the meta-science literature \emph{before} the replication project ($\pi \approx 0.10$, treated as a calibration parameter consistent with scenarios in \citealp{ioannidis2005}; median power $\approx 0.35$, \citealp{cohen1962,sedlmeier1989}) imply $\PPV \approx 0.44$ and a replication rate of approximately 36\%. These parameters were not calibrated to match the Open Science Collaboration's finding; the consistency is a structural implication of the framework applied to independently documented conditions.

\subsection{The Architectural Argument}

At its core, the paper shifts explanation from \emph{behavioral} to \emph{architectural}. Behavioral explanations locate the problem in researchers' choices: p-hacking, selective reporting, and incentive-responsive behavior. Architectural explanations locate it in the structure of the research system itself. In this paper, ``architectural'' means determined by field-level operating parameters ($\pi$, $\alpha$, $1-\beta$) and publication decision rules, rather than by idiosyncratic researcher behavior. The distinction matters because the two explanations prescribe different remedies and carry different moral implications.

An architectural explanation does not exonerate misconduct. Questionable research practices are real, and the specification search collapse (Section~\ref{sec:collapse}) formalizes their damage. But the architectural framing makes explicit what behavioral explanations can obscure: even a hypothetical field populated entirely by honest, careful, technically sophisticated researchers cannot, at these operating parameters, push PPV above 69\% at $\pi = 0.10$ and $\alpha = 0.05$.

\citet{ioannidis2005} showed that PPV is low under realistic assumptions. The present paper shows that, under fixed $\alpha$ and a binary significance-based publication architecture, PPV may be structurally bounded far below common reliability targets, a constraint of the architecture rather than a contingent failing. Where Ioannidis diagnosed a snapshot (the PPV at a given parameter configuration), the present framework diagnoses a trajectory: the generational dynamics, field lifetime, and degenerative programme criterion characterize how reliability evolves as fields mature and follow-up research accumulates. The PPV identity underlying the Certainty Bound has appeared in earlier work on false-positive rates (see, e.g., \citealp{wacholder2004}); the contributions here are design-level interpretation, thresholding, collapse analysis, and institutional implications.

A scope clarification is warranted. The framework addresses the reliability of \emph{binary significance claims}: the posterior probability that a statistically significant finding reflects a genuine effect. Continuous estimation, Bayesian inference, and richer likelihood-based approaches can extract more information from data. The present analysis characterizes the ceiling that binary publication architectures impose; these constraints arise from publication architecture, not from statistical theory itself.

The paper makes four contributions. First, it proves a structural infeasibility result for high reliability below a critical prior $\pcrit$ and identifies two mechanistically distinct collapse routes (observational confounding and specification search) that share the invariant $\Leff \to 1$. Second, it derives a heterogeneity tax via Jensen's inequality and a replication bridge linking the framework to observed replication rates. Third, it develops generational dynamics yielding a quantitative degenerative programme criterion. Fourth, it proves the Replication Pipeline Theorem, establishing the replication pipeline as the natural unit of evidence in low-prior settings.

Throughout, the aim is not to replace existing metascience tools, but to supply a structural criterion for when high reliability is mathematically feasible in the first place.

An interactive implementation of the Certainty Bound diagnostic, pipeline calculator, and reliability landscape is available at \url{https://mpollanen.github.io/certainty-bound-tool/}.\footnote{The tool is also archived on OSF: \url{https://osf.io/c5wun}.}

\subsection{Roadmap}

Section~\ref{sec:framework} develops the formal framework, Section~\ref{sec:majority} proves the Majority-False and Cost of Discovery Theorems, and Section~\ref{sec:replication} establishes the Replication Bridge and its consistency with the OSC finding. Sections~\ref{sec:collapse}--\ref{sec:escape} analyze collapse mechanisms and escape routes, including the Replication Pipeline Theorem. Sections~\ref{sec:dynamics}--\ref{sec:design} develop field dynamics, the reliability landscape, and design requirements for reliable research; Sections~\ref{sec:discussion} and~\ref{sec:conclusion} draw out implications and connections to the philosophy of science.

%==========================================================
% SECTION 2: THE CERTAINTY BOUND
%==========================================================
\section{The Certainty Bound: Formal Framework}
\label{sec:framework}

\subsection{Setup}

\begin{definition}[State Space and Prior]
\label{def:state}
Let $H \in \{0, 1\}$ denote the truth value of a hypothesis, where $H = 1$ means the effect exists and $H = 0$ means it does not. Let $\pi := \Prob(H = 1)$ denote the \emph{prior probability} that a hypothesis selected for investigation is true.
\end{definition}

The prior $\pi$ is not the fraction of all conceivable hypotheses that are true; it is the probability that a hypothesis selected for testing reflects a genuine effect, given the theoretical reasoning, preliminary data, and incentive structures that led to its selection. Clinical trials with Phase~II evidence might have $\pi \approx 0.30$; genome-wide scans, $\pi \approx 10^{-5}$; exploratory social psychology experiments, $\pi \approx 0.05$--$0.15$. The prior is therefore a property of a field's hypothesis-generation process, summarizing how ambitious or conservative its research agenda is.

\begin{definition}[Test Outcomes and Error Rates]
A statistical test produces a binary outcome: significant ($T = 1$) or non-significant ($T = 0$). The test is characterized by the significance level $\alpha := \Prob(T = 1 \mid H = 0)$ and power $1-\beta := \Prob(T = 1 \mid H = 1)$.
\end{definition}

\begin{remark}[Nominal versus Effective Error Rates]
\label{rem:nominal_effective}
We distinguish \emph{nominal} from \emph{effective} error rates. The effective $\aeff$ may exceed the nominal $\alpha$ due to analytical flexibility or selective reporting. When we write $\alpha$ and $1-\beta$ without subscripts, we mean the rates governing the actual acceptance decision $T$; these coincide with nominal rates when no specification search or selective reporting is present. Unless explicitly marked (e.g., $\aeff$, $\Leff$), $\alpha$ and $\beta$ refer to nominal rates. Sections~\ref{sec:framework}--\ref{sec:replication} work with nominal parameters; Section~\ref{sec:collapse} introduces effective rates via specification search and confounding. The PPV identity~\eqref{eq:ppv} applies to whichever rates govern the actual decision process.
\end{remark}

\paragraph{Model Assumptions.}
Throughout the paper we assume $\pi \in (0,1)$, $\alpha \in (0,1)$, $\beta \in [0,1)$, and $\Lambda = (1-\beta)/\alpha \in (0,\infty)$. All probability statements concern the decision rule used to generate the binary claim $T$; that is, $\alpha$ and $1-\beta$ are operating characteristics of the actual testing pipeline. Conditional-independence assumptions will be stated explicitly whenever products or powers of $\Lambda$ are taken.

\begin{definition}[Experimental Leverage]
\label{def:leverage}
The \emph{experimental leverage} of a statistical test is
\[
\Lambda := \frac{1-\beta}{\alpha} = \frac{\text{power}}{\text{false-positive rate}}.
\]
Leverage measures how much more likely a significant result is when the hypothesis is true than when it is false. A test with power 0.80 and $\alpha = 0.05$ has $\Lambda = 16$.
\end{definition}

\begin{remark}[Minimal Discrimination and Sign Reversal]
\label{rem:min_discrimination}
The threshold $\Lambda = 1$ (equivalently, $1-\beta = \alpha$) is the boundary between evidence and anti-evidence for a significant result. From~\eqref{eq:ppv}, if $\Lambda = 1$ then $\PPV = \pi$, so a significant result carries no information beyond the prior. If $\Lambda < 1$, then $\PPV < \pi$, so significance is anti-evidential relative to the prior. The framework remains valid in this regime, but the direction of evidential update reverses. Many threshold results below therefore assume minimal discrimination, $\Lambda > 1$.
\end{remark}

\begin{definition}[Positive Predictive Value]
$\PPV := \Prob(H = 1 \mid T = 1)$.
\end{definition}

\subsection{The Certainty Bound}

\begin{theorem}[Certainty Bound]
\label{thm:certainty_bound}
For any statistical test with leverage $\Lambda$ applied to hypotheses with prior $\pi \in (0,1)$:
\begin{equation}
\label{eq:ppv}
\PPV = \frac{\pi\Lambda}{\pi\Lambda + (1-\pi)}.
\end{equation}
Equivalently, in log-odds form:
\begin{equation}
\label{eq:logodds}
\log\frac{\PPV}{1-\PPV} = \log\frac{\pi}{1-\pi} + \log\Lambda.
\end{equation}
Achieving $\PPV \geq \tau$ requires
\begin{equation}
\label{eq:lambda_req}
\Lambda \geq \Lreq(\tau, \pi) := \frac{\tau}{1-\tau} \cdot \frac{1-\pi}{\pi}.
\end{equation}
\end{theorem}

\begin{proof}
By Bayes' theorem: $\PPV = (1-\beta)\pi/[(1-\beta)\pi + \alpha(1-\pi)] = \pi\Lambda/[\pi\Lambda + (1-\pi)]$. Taking log-odds gives~\eqref{eq:logodds}. Setting $\PPV \geq \tau$ and solving yields~\eqref{eq:lambda_req}.
\end{proof}

The log-odds form makes the informational content transparent: once evidence is reduced to a significance decision, the test contributes exactly $\log\Lambda$ to posterior log-odds, and nothing more. Statistical significance is not a property of the evidence itself but of a decision rule. Within this architecture, the maximum evidential contribution of a significance decision is determined by the protocol's leverage. Researcher diligence matters insofar as it changes the operating parameters, including effective error rates, power, identification, or the priors of tested hypotheses, but it cannot extract more from a given binary decision than $\log\Lambda$ provides.

The Certainty Bound is exact for binary decisions ($T \in \{0,1\}$). It does not bound what is attainable using richer evidence summaries (full likelihoods, posterior distributions, or continuous effect-size estimates), which can extract more information from the same data.

The framework links, but does not equate, three quantities: PPV (reliability of significant findings given a binary publication filter), the replication rate under a specified replication design (Section~\ref{sec:replication}), and the latent prior probability $\pi$ of tested hypotheses. Conflating these is a common source of confusion; each enters the analysis in a distinct role.

\begin{example}[Required Leverage]
\label{ex:leverage}
For target $\tau = 0.95$, the odds ratio $\tau/(1-\tau) = 19$:
\begin{center}
\begin{tabular}{ccc}
\toprule
Prior $\pi$ & Prior odds $(1-\pi)/\pi$ & Required $\Lreq$ \\
\midrule
0.50 & 1 & 19 \\
0.10 & 9 & 171 \\
0.01 & 99 & 1{,}881 \\
$10^{-5}$ & 99{,}999 & $1.9 \times 10^6$ \\
\bottomrule
\end{tabular}
\end{center}
Required leverage scales as $O(1/\pi)$ as the prior decreases.
\end{example}

\subsection{The Fixed-$\alpha$ Ceiling}

\begin{theorem}[Fixed-$\alpha$ Ceiling]
\label{thm:fixed_alpha_ceiling}
For any statistical test with fixed significance level $\alpha \in (0,1)$, regardless of sample size:
\begin{equation}
\label{eq:fixed_ceiling}
\PPV \leq \PPV^{\mathrm{ceil}}(\pi, \alpha) := \frac{\pi}{\pi + \alpha(1-\pi)}.
\end{equation}
This bound is tight, achieved in the limit as power $\to 1$.
\end{theorem}

\begin{proof}
Since $\PPV$ is strictly increasing in $\Lambda$ (with $\frac{\partial}{\partial\Lambda}\PPV = \pi(1-\pi)/[\pi\Lambda + (1-\pi)]^2 > 0$), and $\Lambda = (1-\beta)/\alpha \leq 1/\alpha$ for $\beta \geq 0$, with equality as $\beta \to 0$: $\PPV^{\mathrm{ceil}} = \pi/[\pi + \alpha(1-\pi)]$.
\end{proof}

\begin{corollary}[Ceiling Values at $\alpha = 0.05$]
\label{cor:ceiling_values}
\begin{center}
\begin{tabular}{ccc}
\toprule
Prior $\pi$ & $\PPV^{\mathrm{ceil}}$ & Interpretation \\
\midrule
0.50 & 95.2\% & Barely meets 95\% target \\
0.30 & 89.6\% & Cannot reach 90\% PPV \\
0.10 & 68.9\% & Cannot exceed 69\% PPV \\
0.05 & 51.3\% & Barely majority-true \\
0.01 & 16.8\% & Predominantly false \\
\bottomrule
\end{tabular}
\end{center}
\end{corollary}

For pre-reform psychology ($\pi \approx 0.10$, $\alpha = 0.05$), the maximum attainable PPV with any sample size is 69\%, far below the 95\% reliability often assumed in interpretation. The ceiling depends only on $\alpha$ and $\pi$; no increase in sample size can move it.

\subsection{The Critical Prior and the Infeasibility Index}

\begin{definition}[Infeasibility Index]
\label{def:tension}
The \emph{infeasibility index} is
\begin{equation}
\label{eq:psi}
\Psi(\tau, \pi, \alpha, \beta) := \frac{\Lreq(\tau, \pi)}{\Lambda(\alpha, \beta)} = \frac{\tau}{1-\tau} \cdot \frac{1-\pi}{\pi} \cdot \frac{\alpha}{1-\beta}.
\end{equation}
\end{definition}

\begin{theorem}[Critical Prior and Structural Infeasibility]
\label{thm:critical_prior}
For given $\alpha \in (0,1)$, $\beta \in [0,1)$, and target $\tau \in (0,1)$:
\begin{enumerate}[label=(\alph*)]
\item \textbf{Feasibility:} The target PPV is achievable if and only if $\pi \geq \pcrit(\tau, \alpha, \beta)$, where
\begin{equation}
\label{eq:pi_crit}
\pcrit(\tau, \alpha, \beta) = \frac{\tau\alpha}{(1-\tau)(1-\beta) + \tau\alpha}.
\end{equation}
\item \textbf{Infeasibility:} If $\Psi > 1$ (equivalently $\pi < \pcrit$), the target PPV is structurally unattainable at the stated operating parameters. No increase in sample size alone under fixed $\alpha$ can exceed the Fixed-$\alpha$ Ceiling, and no improvement in study conduct that leaves the operating parameters unchanged can overcome the constraint.
\end{enumerate}
\end{theorem}

\begin{proof}
Part (a): By Theorem~\ref{thm:certainty_bound}, $\PPV \geq \tau$ if and only if $\Lambda \geq \Lreq(\tau,\pi)$. Substituting $\Lambda = (1-\beta)/\alpha$ and solving for $\pi$ yields~\eqref{eq:pi_crit}.

Part (b): By Definition~\ref{def:tension}, $\Psi > 1$ if and only if $\Lreq(\tau,\pi) > \Lambda$, equivalently (by part (a)) if and only if $\pi < \pcrit(\tau,\alpha,\beta)$. In that case, $\PPV < \tau$ at the stated operating parameters. Moreover, under fixed $\alpha$, Theorem~\ref{thm:fixed_alpha_ceiling} implies $\PPV \leq \pi/[\pi+\alpha(1-\pi)]$, so increasing sample size alone cannot overcome the constraint once the target lies above the ceiling. Any intervention that leaves $(\pi,\alpha,\beta)$ unchanged also leaves $\Lambda$ unchanged, and therefore cannot change the implied PPV.
\end{proof}

The theorem reframes the question: instead of asking whether PPV is high enough, it asks whether the field prior exceeds $\pcrit$. The locus of analysis shifts from individual study quality to the operating parameters of the research enterprise.

\begin{table}[htbp]
\centering
\caption{Critical prior $\pcrit(\tau, \alpha, \beta)$ for various configurations at target $\tau = 0.95$.}
\label{tab:critical_priors}
\begin{tabular}{cccc}
\toprule
$\alpha$ & Power & $\pcrit$ & Implication \\
\midrule
0.05 & 0.80 & 54.3\% & Majority of tested hypotheses must be true \\
0.05 & 0.50 & 65.5\% & Two in three must be true \\
0.05 & 0.30 & 76.0\% & Three in four must be true \\
0.01 & 0.80 & 19.2\% & One in five must be true \\
0.005 & 0.80 & 10.6\% & One in nine must be true \\
$5\times10^{-8}$ & 0.80 & $1.2\times10^{-6}$ & One in 840{,}000 \\
\bottomrule
\end{tabular}
\end{table}

Three patterns stand out. Reducing $\alpha$ from 0.05 to 0.005 lowers $\pcrit$ from 54\% to 11\%, a fivefold improvement. Reducing power raises $\pcrit$: power of 0.30 rather than 0.80 demands that three-quarters of tested hypotheses be true. The GWAS threshold ($\alpha = 5\times10^{-8}$; see \citealp{peer2008}) reduces $\pcrit$ to about one in 840{,}000, enabling reliable discovery at extremely low priors.

\subsection{Heterogeneity Tax}

\begin{proposition}[Prior Heterogeneity Penalty]
\label{prop:jensen}
Let $\Pi$ be a non-degenerate random variable on $(0,1)$ with mean $\bar\pi$ and variance $\sigma^2_\pi > 0$. For fixed $\Lambda > 1$:
\begin{equation}
\label{eq:jensen}
\E[\PPV(\Pi)] < \PPV(\bar\pi),
\end{equation}
with the gap approximated by $\Lambda(\Lambda-1)\sigma^2_\pi / [\bar\pi(\Lambda-1)+1]^3$. Strict inequality holds when $\Lambda > 1$ and $\Pi$ is non-degenerate.
\end{proposition}

\begin{proof}
See Appendix~\ref{app:jensen}.
\end{proof}

Heterogeneity imposes a tax on reliability. If a field mixes hypothesis classes with different priors, such as exploratory fishing alongside confirmatory follow-ups, its \emph{average} PPV falls below the PPV implied by its \emph{average} prior, even when leverage is held fixed. The mechanism is concavity, not bias.

The penalty is concrete: at $\Lambda = 16$, $\PPV(0.10) = 0.64$. But mixing $\pi \in \{0.02, 0.18\}$ with equal probability gives average PPV $\approx 0.51$, despite the same mean prior. Concavity penalizes low-$\pi$ observations more than it rewards high-$\pi$ ones, so the low-prior half of a mixed field drags the average down more than the high-prior half lifts it. Fields that adopt standardized protocols and registered reports reduce this tax; those mixing exploratory and confirmatory research pay it in full.

This yields a testable implication: holding leverage approximately fixed, fields with greater prior heterogeneity should exhibit lower replication rates than single-prior calibrations predict.

\begin{remark}[Aggregation Bias]
Treating a discipline like ``Psychology'' as a monolithic block with a single prior $\pi \approx 0.10$ is a simplification; in reality, subfields vary widely. However, because the PPV function is strictly concave (Proposition~\ref{prop:jensen}), this aggregation introduces an optimistic bias. A field composed of a mix of high-reliability and low-reliability subfields will have a lower aggregate replication rate than a homogeneous field operating at the mean parameters. Consequently, the infeasibility results presented here should be interpreted as a \emph{conservative} upper bound on field-wide reliability: the structural reality is likely strictly worse than the aggregated model suggests.
\end{remark}

\begin{figure}[htbp]
\centering
\begin{tikzpicture}
\begin{axis}[
    width=12cm, height=8cm,
    xlabel={Prior Probability $\pi$},
    ylabel={Positive Predictive Value},
    xmin=0, xmax=0.7, ymin=0, ymax=1.05,
    xtick={0, 0.059, 0.106, 0.20, 0.40, 0.543, 0.60},
    xticklabels={0, 0.06, 0.11, 0.20, 0.40, 0.54, 0.60},
    xticklabel style={font=\scriptsize},
    ytick={0, 0.25, 0.50, 0.75, 0.95},
    yticklabels={0, 0.25, 0.50, 0.75, 0.95},
    yticklabel style={font=\scriptsize},
    legend style={at={(0.98,0.35)}, anchor=east, font=\small},
    axis lines=left,
    grid=major,
    clip=false,
]
\fill[red, opacity=0.10] (axis cs:0,0) rectangle (axis cs:0.543,1.05);
\fill[red, opacity=0.25] (axis cs:0,0) rectangle (axis cs:0.059,1.05);
\addplot[thick, blue, domain=0.001:0.7, samples=100]
    {(x*16)/(x*16 + (1-x))};
\addlegendentry{$\alpha=0.05$, power$=0.80$}
\addplot[thick, red, domain=0.001:0.7, samples=100]
    {(x*160)/(x*160 + (1-x))};
\addlegendentry{$\alpha=0.005$, power$=0.80$}
\addplot[thick, black, dashed, domain=0:0.7] {0.95};
\addlegendentry{Target $\PPV = 0.95$}
\addplot[thick, black, dotted, domain=0:0.7] {0.50};
\addlegendentry{Majority-false boundary}
\draw[thick, blue, dotted] (axis cs:0.543,0) -- (axis cs:0.543,0.95);
\draw[thick, red, dotted] (axis cs:0.106,0) -- (axis cs:0.106,0.95);
    \draw[thick, gray!60, dashed] (axis cs:0.02,0.246) -- (axis cs:0.18,0.778);
    \filldraw[gray!60] (axis cs:0.02,0.246) circle (1.5pt);
    \filldraw[gray!60] (axis cs:0.18,0.778) circle (1.5pt);
\node[font=\scriptsize, rotate=90] at (axis cs:0.03, 0.55) {\shortstack{Majority\\false}};
\node[font=\scriptsize] at (axis cs:0.32, 0.08) {Infeasible at $\alpha=0.05$};
\node[font=\scriptsize] at (axis cs:0.62, 0.08) {Feasible};
\end{axis}
\end{tikzpicture}
\caption{PPV as a function of prior probability under two significance thresholds. The darkly shaded region ($\pi < 0.059$) is the majority-false regime at conventional parameters. The lightly shaded region ($0.059 < \pi < 0.543$) is infeasible for target $\tau = 0.95$ at $\alpha = 0.05$. Tightening $\alpha$ to 0.005 shifts the critical prior to 0.106. The chord between $\pi = 0.02$ and $\pi = 0.18$ lies below the curve, visualizing the heterogeneity tax.}
\label{fig:certainty_curve}
\end{figure}

\begin{remark}[Independence]
\label{rem:independence}
Several results below (the Replication Pipeline Theorem, the specification search model, and the Cost of Discovery Theorem) treat studies or repeated tests as statistically independent. This is an idealized benchmark. Dependence across studies (shared samples, correlated hypotheses, or sequential updating) typically weakens leverage multiplication and strengthens the practical force of the infeasibility results.
\end{remark}

\paragraph{Assumption Map.}
For clarity, we distinguish the assumptions used by different results. The Certainty Bound (Theorem~\ref{thm:certainty_bound}) and Fixed-$\alpha$ Ceiling (Theorem~\ref{thm:fixed_alpha_ceiling}) require only binary coarsening and the stated error rates. The Replication Pipeline Theorem (Theorem~\ref{thm:rep_leverage}) and Cost of Discovery (Theorem~\ref{thm:cost}) additionally assume conditional independence of studies given $H$. The Observational Collapse (Theorem~\ref{thm:obs_collapse}) is an asymptotic result in sample size under persistent confounding. The Specification Search Collapse (Proposition~\ref{prop:spec_collapse}) uses sequential independence of specification attempts. The Generational Dynamics (Theorem~\ref{thm:generational_decay}) and Field Lifetime (Theorem~\ref{thm:field_lifetime}) are stylized dynamic models with fixed $\Lambda$ and deterministic transmission.

%==========================================================
% SECTION 3: IMPOSSIBILITY REGIMES
%==========================================================
\section{Two Impossibility Regimes}
\label{sec:majority}

\subsection{The Majority-False Theorem}

\begin{theorem}[Majority-False Threshold]
\label{thm:majority_false}
For a test with leverage $\Lambda$, the majority of significant findings are false ($\PPV < 1/2$) if and only if
\begin{equation}
\label{eq:majority_false}
\pi < \pi_{1/2} := \frac{1}{1 + \Lambda} = \frac{\alpha}{(1-\beta) + \alpha}.
\end{equation}
At $\alpha = 0.05$ and power $0.80$: $\pi_{1/2} = 1/17 \approx 5.9\%$. At $\alpha = 0.05$ and power $0.35$: $\pi_{1/2} = 12.5\%$.
\end{theorem}

\begin{proof}
$\PPV < 1/2$ iff $\pi\Lambda < (1-\pi)$ iff $\pi < 1/(1+\Lambda)$.
\end{proof}

In exploratory research, the majority-false threshold is easy to cross. A field in which fewer than roughly one in eight to seventeen hypotheses is genuinely true (depending on power) will produce a literature where the majority of significant findings are false, without any questionable research practices.

\subsection{The Cost of Discovery}

\begin{theorem}[Cost of Discovery]
\label{thm:cost}
In a field with constant parameters $(\pi, \alpha, \beta)$ and independent studies, the long-run expected ratio of false positives to true positives (equivalently, the ratio of expected counts over $N$ studies as $N\to\infty$) is
\begin{equation}
\label{eq:waste}
W = \frac{(1-\pi)\alpha}{\pi(1-\beta)} = \frac{1-\PPV}{\PPV}.
\end{equation}
\end{theorem}

\begin{proof}
See Appendix~\ref{app:cost}.
\end{proof}

\begin{example}[Scientific Cost Across Fields]
\label{ex:waste}
\begin{center}
\begin{tabular}{lccccc}
\toprule
\textbf{Field} & $\pi$ & Power & $\alpha$ & \textbf{PPV} & \textbf{$W$} \\
\midrule
Candidate genes & 0.02 & 0.50 & 0.05 & 17\% & 4.9 \\
Pre-reform psychology & 0.10 & 0.35 & 0.05 & 44\% & 1.3 \\
Nutritional epidemiology & 0.08 & 0.60 & 0.05 & 51\% & 1.0 \\
Well-powered RCT & 0.30 & 0.80 & 0.05 & 87\% & 0.15 \\
GWAS & $10^{-5}$ & 0.80 & $5\times10^{-8}$ & 99.4\% & 0.006 \\
\bottomrule
\end{tabular}
\end{center}
The candidate gene literature produced nearly 5 false positives for every true positive and required roughly 100 studies per genuine discovery. Pre-reform psychology required approximately 29 studies per genuine discovery.
\end{example}

%==========================================================
% SECTION 4: REPLICATION BRIDGE
%==========================================================
\section{The Replication Bridge}
\label{sec:replication}

\subsection{The Replication Bridge Theorem}

\begin{theorem}[Replication Bridge]
\label{thm:replication_bridge}
Consider an original study with positive predictive value $\PPV_o$ and an independent replication conducted at level $\alpha_r$ with power $1-\beta_r$. The probability of a significant replication given a significant original is:
\begin{equation}
\label{eq:rep_bridge}
\Prob(T_r = 1 \mid T_o = 1) = \PPV_o \cdot (1-\beta_r) + (1 - \PPV_o) \cdot \alpha_r.
\end{equation}
Conversely, given an observed replication rate $R$ and provided $1-\beta_r > \alpha_r$:
\[
\widehat{\PPV} = \frac{R - \alpha_r}{(1-\beta_r) - \alpha_r}.
\]
\end{theorem}

\begin{proof}
By the law of total probability, conditioning on $H$ and using independence of $T_r$ and $T_o$ conditional on $H$:
\begin{align*}
\Prob(T_r = 1 \mid T_o = 1) &= \Prob(H=1 \mid T_o=1)\,\Prob(T_r=1 \mid H=1) + \Prob(H=0 \mid T_o=1)\,\Prob(T_r=1 \mid H=0) \\
&= \PPV_o \cdot (1-\beta_r) + (1 - \PPV_o) \cdot \alpha_r.
\end{align*}
Solving for $\PPV_o$ yields the inversion formula.
\end{proof}

The Replication Bridge assumes independence of original and replication conditional on $H$; shared materials, experimenters, or participant populations may violate this, attenuating realized leverage multiplication.

\subsection{Consistency with the Observed Replication Rate}

\begin{mdframed}
\textbf{Box 1: Parameter Calibration for Pre-Reform Psychology.}

\smallskip
The following estimates were established by independent sources before the Open Science Collaboration's replication project (2015).

\smallskip
\textbf{Statistical power ($1-\beta \approx 0.35$):} \citet{cohen1962} documented low typical power in abnormal-social psychology. In a follow-up historical comparison of the same journal, \citet{sedlmeier1989} reported persistently low median power for a medium effect (0.46 in 1960 versus 0.37 in 1984), with median power as low as 0.25 in cases where nonsignificance was interpreted as confirmation of the null. \citet{button2013} report a median statistical power of about 21\% in neuroscience and summarize a broader range of roughly 8--31\% across analyses and subfields. We adopt $1-\beta = 0.35$ as a central estimate for pre-reform social psychology.

\smallskip
\textbf{Prior probability ($\pi \approx 0.10$, calibration parameter):} The field-level prior is not directly observable and is treated as a calibration parameter summarizing the base rate of genuinely true hypotheses entering formal tests. We use $\pi = 0.10$ as a central value for pre-reform social psychology, consistent with scenarios explored in \citet{ioannidis2005}, with sensitivity analysis over $\pi \in [0.05, 0.20]$.

\smallskip
\textbf{Significance level:} $\alpha = 0.05$, the conventional threshold in much of the literature.

\smallskip
These values imply $\Lambda = 0.35/0.05 = 7$ and $\PPV = 0.10 \times 7/(0.10 \times 7 + 0.90) = 0.44$.
\end{mdframed}

Substituting the parameters from Box~1 into the Replication Bridge, with replication power $1-\beta_r \approx 0.75$ (based on the OSC's design of replications to achieve high power; see \citealp{osc2015}):\footnote{Winner's curse \citep{ioannidis2008} compounds this effect: if replication teams power their studies for the published effect size $\hat\theta > \theta$, realized replication power falls below the intended level. The predicted replication rate of 36--38\% should therefore be interpreted as an approximate upper bound under these parameter assumptions.}

\[
\Prob(\text{rep.\ significant} \mid \text{orig.\ significant}) = 0.44 \times 0.75 + 0.56 \times 0.05 = 0.330 + 0.028 = 0.358.
\]

With $1-\beta_r = 0.80$: $\Prob = 0.44 \times 0.80 + 0.56 \times 0.05 = 0.380$.

The implied replication rate of 36--38\% is consistent with the Open Science Collaboration's observed 36\% \citep{osc2015}. A more precise characterization is \emph{retrodiction}: the framework is applied to parameter ranges established independently, and the implied rate is checked against a subsequently observed value. Bridge inversion provides a complementary check: from $R = 0.36$ and $1-\beta_r = 0.75$, we recover $\widehat{\PPV} = (0.36-0.05)/(0.75-0.05) = 0.44$, consistent with the parameter-derived estimate. This agreement is algebraic consistency within the framework, not independent confirmation; its value lies in showing that the bridge is internally coherent at the observed replication rate.

The retrodiction does not identify a unique parameterization. Multiple $(\pi, 1-\beta)$ combinations can produce similar replication rates (Table~\ref{tab:sensitivity}), and bridge inversion identifies implied PPV, not $\pi$ alone, absent assumptions about original power and effective $\alpha$. The key point is regime-level robustness: across a broad range of independently plausible pre-reform parameters, the field remains below the target-reliability feasibility boundary.

\begin{table}[htbp]
\centering
\caption{Implied replication rates across prior probability and power, at $\alpha = 0.05$, $\alpha_r = 0.05$, $1-\beta_r = 0.75$. The observed OSC rate of 36\% is shown in bold.}
\label{tab:sensitivity}
\begin{tabular}{lccccc}
\toprule
 & \multicolumn{5}{c}{Prior probability $\pi$} \\
\cmidrule(l){2-6}
Power $1-\beta$ & 0.05 & 0.10 & 0.15 & 0.20 & 0.30 \\
\midrule
0.25 & 20\% & 30\% & 38\% & 44\% & 53\% \\
0.35 & 24\% & \textbf{36\%} & 44\% & 50\% & 58\% \\
0.50 & 29\% & 42\% & 50\% & 55\% & 62\% \\
\bottomrule
\end{tabular}
\end{table}

Across all fifteen cells, $\Psi > 1$ for a target of $\tau = 0.95$: the infeasibility finding is robust to substantial parameter uncertainty.

Additional consistency checks are broadly consistent with the framework. \citet{camerer2016} found approximately 61\% replication in laboratory economics; \citet{camerer2018} found approximately 62\% for social science experiments in \emph{Nature} and \emph{Science}. These higher rates are consistent with higher priors or higher leverage. \citet{errington2021} reported replication success rates that varied by criterion in preclinical cancer biology; bridge inversion (assuming $\alpha_r = 0.05$ and replication power $\approx 0.80$) recovers $\widehat{\PPV}$ in the range consistent with the infeasible but majority-true regime.

%==========================================================
% SECTION 5: TWO ROUTES TO COLLAPSE
%==========================================================
\section{Two Routes to Collapse}
\label{sec:collapse}

Sufficient leverage is necessary for reliability, but leverage can be destroyed. Two recurrent features of scientific practice drive leverage downward and can, in the collapse limit, return PPV to the prior: persistent confounding in observational data and analytical flexibility under publication pressure. Both mechanisms can arise under standard disciplinary practice and do not require bad faith. The collapse results are asymptotic and mechanism-specific; they do not imply that all observational research or all specification flexibility collapses to the prior in practice.

\subsection{Route I: Observational Collapse}

\begin{definition}[Causal versus Associational Null]
\label{def:null_types}
Let $H_{\mathrm{causal}} = 0$ denote no causal effect and $H_{\mathrm{assoc}} = 0$ denote no association after adjustment for measured covariates. With unmeasured confounders of magnitude $b > 0$, the causal null does not imply the associational null: a spurious association of magnitude $b$ remains.
\end{definition}

In the theorem below, $z_\alpha$ denotes the upper $\alpha$-tail standard normal critical value, i.e., $\Prob(Z > z_\alpha)=\alpha$ for $Z\sim N(0,1)$.

\begin{theorem}[Observational Collapse]
\label{thm:obs_collapse}
Consider an observational study testing for a causal effect via $Z_n = \hat\theta_n / \mathrm{SE}(\hat\theta_n)$, with one-sided rejection for large positive values (the two-sided analogue is stated in part (d)), under the following conditions:
\begin{enumerate}[label=(\roman*)]
\item Persistent unmeasured confounding: there exists a bias term $b_n$ in the tested direction such that under $H_{\mathrm{causal}}=0$, $Z_n$ is asymptotically $N(\sqrt{n}\,b_n/\sigma, 1)$ for some $\sigma > 0$, with $\sqrt{n}\,b_n \to +\infty$;
\item The identification strategy is not strengthened with $n$ in a way that forces $b_n \to 0$ fast enough to prevent $\sqrt{n}\,b_n \to \infty$;
\item Consistent test under $H_{\mathrm{causal}}=1$ ($1-\beta_n \to 1$).
\end{enumerate}
Then as $n \to \infty$:
\begin{enumerate}[label=(\alph*)]
\item $\aeff(n) \to 1$;
\item $\Lambda_{\mathrm{eff}}(n) \to 1$;
\item $\PPV_n \to \pi$;
\item For two-sided rejection ($|Z_n| > z_{\alpha/2}$), if $\sqrt{n}\,|b_n| \to \infty$, then $\aeff(n) \to 1$ as well.
\end{enumerate}
\end{theorem}

\begin{proof}
See Appendix~\ref{app:obs_collapse}.
\end{proof}

Under the stated assumptions, the Observational Collapse is qualitatively worse than the Fixed-$\alpha$ Ceiling. The ceiling for a randomized study at $\pi = 0.10$ and $\alpha = 0.05$ is 69\%; the asymptotic value for a confounded observational study is 10\%, the prior itself. When persistent confounding is present, increasing sample size makes the test progressively less informative about the causal hypothesis. The theorem applies to designs in which unmeasured confounding does not shrink under a fixed identification strategy. If confounding is absent, or if it diminishes through improved design or measurement, the collapse need not occur. If $\sqrt{n}\,b_n \to -\infty$ while testing in the positive direction, then $\aeff(n) \to 0$ rather than~1, so the directionality of confounding relative to the rejection region matters.

\begin{corollary}[Confident Falsehood Under Continuous Estimation]
\label{cor:confident_falsehood}
Under standard regularity conditions for misspecified $M$-estimation (e.g., \citealt{white1982}), if the fitted model omits confounders responsible for bias, then $\hat\theta_n \xrightarrow{p} \theta^\star$, where $\theta^\star$ is the pseudo-true parameter (the Kullback--Leibler minimizer within the misspecified model class); in common confounding settings, $\theta^\star$ coincides with the confounded associational estimand. Moreover, $\mathrm{SE}(\hat\theta_n) \to 0$. Under standard Bayesian misspecification regularity conditions (e.g., \citealt{kleijn2012}), the posterior likewise concentrates near $\theta^\star$. Binary collapse returns the researcher to uncertainty (the prior); continuous-estimation collapse can instead produce high confidence around a structurally incorrect estimand.
\end{corollary}

This asymmetry helps explain an otherwise puzzling feature of certain research traditions: large observational studies converge on precise estimates that experimental evidence later contradicts \citep{ioannidis2018}. The literature accumulates tight confidence intervals centered on the confounding bias rather than honest uncertainty. Sample size reduces uncertainty about $\theta^\star$, not about the causal effect $\theta$.

\begin{corollary}[Identification as Escape]
\label{cor:identification_escape}
Designs achieving valid identification (randomized experiments, instrumental variables, regression discontinuity, difference-in-differences when identifying assumptions hold) eliminate or bound $b$, restoring nominal leverage. Under the assumptions of Theorem~\ref{thm:obs_collapse}, and absent a mechanism that forces $b_n \to 0$ sufficiently fast as $n \to \infty$, valid identification (or an equivalent bias-shrinking mechanism) is necessary for asymptotically reliable causal inference in this framework.
\end{corollary}

\subsection{Route II: Specification Search Collapse}

\begin{definition}[Sequential Specification Search]
\label{def:spec_search}
A researcher has $m \geq 1$ legitimate analytical specifications. Each specification is tested in turn; if significant, the result is reported and the search stops; if non-significant, the null result is published with probability $q \in [0,1]$ or the researcher proceeds to the next specification. The parameter $q$ governs selective continuation: $q = 0$ means never publish a null before exhausting the search; $q = 1$ means no selective continuation.
\end{definition}

\begin{definition}[Effective Error Rates under Specification Search]
\label{def:effective_rates}
Define $\aeff(q,m) := \Prob(\text{a significant result is published} \mid H = 0)$ and $1 - \beta_{\mathrm{eff}}(q,m) := \Prob(\text{a significant result is published} \mid H = 1)$.
\end{definition}

\begin{lemma}[Effective Error Rates under Specification Search]
\label{lem:effective_errors}
Under sequential search with $m$ independent tests at nominal rates $(\alpha, 1-\beta)$, defining $s_0 = (1-\alpha)(1-q)$ and $s_1 = \beta(1-q)$:
\begin{align}
\aeff(q,m) &= \alpha \cdot \frac{1-s_0^m}{1-s_0}, \label{eq:alpha_eff_pub}\\
1 - \beta_{\mathrm{eff}}(q,m) &= (1-\beta) \cdot \frac{1-s_1^m}{1-s_1}. \label{eq:power_eff_pub}
\end{align}
\end{lemma}

\begin{proof}
See Appendix~\ref{app:spec_search}.
\end{proof}

These formulas are exact under the independence benchmark. Independence is used only to obtain closed forms; the collapse endpoints (Proposition~\ref{prop:spec_collapse}) require only that the search can make $\aeff \to 1$ as $m \to \infty$ under $q = 0$. Dependence across specifications may attenuate or amplify the inflation, depending on the correlation structure. Related behavioral models of analytical flexibility include \citet{simmons2011} and the ``garden of forking paths'' analysis of \citet{gelman2014}.

\begin{theorem}[Discrimination Loss under Specification Search]
\label{thm:ppv_adj}
Under specification search with parameters $(q,m)$, the effective leverage is $\Lambda_{\mathrm{eff}}^{\mathrm{pb}} = \Lambda \cdot D(q,m)$, where the superscript ``pb'' denotes the publication-bias setting, and the discrimination loss factor is
\begin{equation}
\label{eq:degradation}
D(q,m) := \frac{1-s_1^m}{1-s_0^m} \cdot \frac{1-s_0}{1-s_1} < 1
\end{equation}
for all $q < 1$, $m \geq 2$, and $1-\beta > \alpha$.
\end{theorem}

\begin{proof}
See Appendix~\ref{app:disc_loss}.
\end{proof}

\begin{proposition}[Saturation and Collapse]
\label{prop:spec_collapse}
Assume $\Lambda > 1$.
\begin{enumerate}[label=(\alph*)]
\item \textbf{Saturation} (fixed $q \in (0,1)$, $m \to \infty$): Leverage saturates at $\Lambda_{\mathrm{sat}}(q) < \Lambda$; PPV converges to a value strictly above $\pi$ but below the no-bias PPV.
\item \textbf{Unbounded-search collapse} ($q = 0$, $m \to \infty$): $\Lambda_{\mathrm{eff}}^{\mathrm{pb}} \to 1$ and $\PPV \to \pi$.
\end{enumerate}
\end{proposition}

\begin{proof}
See Appendix~\ref{app:saturation}.
\end{proof}

Most real fields likely operate in the saturation regime ($q > 0$, finite $m$) rather than the literal collapse limit. Collapse is a limiting case that reveals the endpoint of the mechanism; saturation alone can place a field in the infeasible regime by depressing $\Leff$ below the required threshold.

\subsection{The Pre-Registration Corollary}

\begin{corollary}[Pre-Registration]
\label{cor:prereg}
Under idealized compliance, pre-registration binds the primary analysis to a single confirmatory specification ($m=1$), restoring $\aeff = \alpha$ regardless of $q$. This eliminates the coupling between publication pressure and specification multiplicity, making PPV of positive findings invariant to file-drawer intensity. In practice, partial compliance or multiple pre-specified analyses can be represented by $m>1$ with constrained search.
\end{corollary}

\subsection{The Double Collapse}
\label{subsec:double_collapse}

\begin{theorem}[Double Collapse]
\label{thm:double_collapse}
Suppose observational confounding and sequential specification search are both present, with specification-search parameters $(m,q)$ fixed in $n$.
\begin{enumerate}[label=(\alph*)]
\item If confounding drives the effective null rejection rate to one, $\aeff(n,m,q)\to 1$, and the effective power under $H=1$ remains consistent, $1-\beta_{\mathrm{eff}}(n,m,q)\to 1$, then $\Leff(n,m,q)\to 1$ as $n\to\infty$.
\item Unbounded search alone ($b=0$, $q=0$, $m\to\infty$) drives $\Leff\to 1$.
\end{enumerate}
\end{theorem}

\begin{proof}
See Appendix~\ref{app:double_collapse}.
\end{proof}

Under part (a), confounding drives $\aeff(n,m,q) \to 1$; specification search may already have pushed it above its nominal level. The Bayes-update factor (the ratio of true-positive to false-positive rates) then becomes a ratio of two quantities converging to~1, so the quotient converges to~1 as well. The two mechanisms shrink the Bayes update in sequence: specification search inflates $\alpha$; confounding erases whatever leverage remains.

\begin{example}[Double Collapse: Numerical Illustration]
Let $\pi = 0.10$, power $= 0.80$, nominal $\alpha = 0.05$, so $\Lambda = 16$ and $\PPV \approx 0.64$. Specification search inflates $\aeff$ to $0.20$, reducing effective leverage to $4$ and $\PPV$ to approximately $0.31$. Adding observational collapse pushes $\Leff$ toward $1$; as $\Leff \to 1$, $\PPV \to 0.10$.
\end{example}

This double-collapse configuration (observational data, flexible analysis, strong publication pressure) describes the historical candidate gene literature \citep{border2019} and portions of observational nutritional epidemiology \citep{ioannidis2018}. Single-channel reforms are insufficient in the double-collapse configuration: better covariate adjustment does not prevent $\aeff$ inflation from specification search, and pre-registration does not manufacture valid identification. Neither route requires bad faith.

%==========================================================
% SECTION 6: ESCAPE MECHANISMS
%==========================================================
\section{Escape from the Infeasible Regime}
\label{sec:escape}

The collapse results also point to the available escape routes. Three mechanisms can repair leverage or move a field toward the feasible regime. Pre-registration (Corollary~\ref{cor:prereg}) enforces nominal $\alpha$, a necessary but generally insufficient step. The remaining two, threshold tightening and replication, are developed here. Of these, replication pipelines are the most broadly applicable, because they bear directly on what constitutes the unit of scientific evidence.

\subsection{The Adaptive Escape in Randomized Experiments}

\begin{theorem}[Adaptive Escape]
\label{thm:adaptive_escape}
Consider a one-sided Gaussian $z$-test for a simple alternative with effect size $\theta_1 > 0$ and known standard error scale $\sigma/\sqrt{n}$, with $\alpha_n = 1-\Phi(c\sqrt{n})$ for some constant $0 < c < \theta_1/\sigma$. Then $\alpha_n \to 0$ exponentially, $(1-\beta_n) \to 1$, $\Lambda_n \to \infty$, and for any $\pi \in (0,1)$ and $\tau \in (0,1)$, $\PPV_n \geq \tau$ for all sufficiently large $n$. The required sample size grows as $n = O(\log \Lreq)$.
\end{theorem}

\begin{proof}
See Appendix~\ref{app:adaptive_escape}.
\end{proof}

Crucially, the Adaptive Escape requires \emph{jointly} increasing $n$ and tightening $\alpha_n$; sample size alone, under fixed $\alpha = 0.05$, cannot push PPV above the Fixed-$\alpha$ Ceiling. The clearest empirical instance is GWAS, which implements Bonferroni correction for $\sim 10^6$ variants ($\alpha = 5\times 10^{-8}$; \citealp{peer2008}), raising $\Lambda$ from~16 to $1.6\times 10^7$. The transition from candidate gene research to GWAS is a particularly clear example of a field escaping the infeasible regime through threshold tightening, a methodological discontinuity whose quantitative signature is consistent with a Kuhnian paradigm shift (Section~\ref{sec:landscape}).

\subsection{Replication Pipelines: The Primary Escape Mechanism}

A single study, even a perfectly conducted, pre-registered randomized experiment, is structurally insufficient for achieving the target reliability in any field operating below the critical prior. In such settings, the natural unit of evidence is the \emph{pipeline}.

\begin{theorem}[Replication Pipeline]
\label{thm:rep_leverage}
Assume $\Lambda > 1$. If a claim is accepted only if all $k$ independent pre-registered studies are significant, with each study conducted at the same nominal level $\alpha$ and power $1-\beta$, the combined leverage is
\begin{equation}
\label{eq:rep_leverage}
\Lambda^{(k)} = \left(\frac{1-\beta}{\alpha}\right)^k = \Lambda^k.
\end{equation}
For any target $\tau \in (0,1)$ and prior $\pi \in (0,1)$, there exists $k^*$ such that $\PPV(\Lambda^{(k)}) \geq \tau$ for all $k \geq k^*$.
\end{theorem}

\begin{proof}
Under conditional independence given $H$, the combined false-positive rate is $\alpha^k$ and the combined true-positive rate is $(1-\beta)^k$, so $\Lambda^{(k)} = \Lambda^k$. Since $\Lambda > 1$, we have $\Lambda^k \to \infty$, and therefore $\PPV(\Lambda^{(k)}) \to 1$ by Theorem~\ref{thm:certainty_bound}. Hence, for any target $\tau \in (0,1)$, there exists $k^*$ such that $\PPV(\Lambda^{(k)}) \geq \tau$ for all $k \geq k^*$.
\end{proof}

This result concerns the evidential standard (``accept if and only if all $k$ studies are significant''). If only successful pipelines are selectively published, the effective pipeline-level $\alpha$ is inflated, and the guarantee requires using the effective rates. Under positive dependence across replications, $\alpha^k$ and $(1-\beta)^k$ are no longer exact; leverage multiplication becomes an upper bound.

Leverage multiplication is geometric: two independent replications at $\Lambda = 16$ yield $\Lambda^{(2)} = 256$, sufficient for $\tau = 0.95$ at priors as low as $\pi \approx 7\%$. Three replications yield $\Lambda^{(3)} = 4{,}096$, sufficient at $\pi \approx 0.5\%$. No sample size increase within a single study can achieve this, because the Fixed-$\alpha$ Ceiling binds before the target is reached. The multiplication is exact under conditional independence given $H$ (Remark~\ref{rem:independence}); shared methods or populations attenuate it.

\begin{corollary}[Minimum Pipeline Depth]
\label{cor:pipeline_depth}
If $\Lambda > 1$, the minimum number of studies in an all-significant independent pipeline (counting the original study) required to achieve $\PPV \geq \tau$ is
\[
k^* = \max\!\left\{1,\left\lceil \frac{\log \Lreq(\tau,\pi)}{\log \Lambda} \right\rceil\right\}.
\]
Equivalently, the minimum number of replications beyond the original study is $k^* - 1$.
\end{corollary}

Institutionally, the implication is direct. A field's standard of evidence should be a pre-registered replication pipeline, not a single study with a $p$-value. A single significant result is a component of evidence, not a unit of it. On this account, journals and funders that treat isolated findings as publishable conclusions are institutionalizing insufficient evidence.

%==========================================================
% SECTION 7: FIELD DYNAMICS
%==========================================================
\section{Field Dynamics}
\label{sec:dynamics}

This section develops stylized dynamic extensions that embed the Certainty Bound in field-level evolutionary models. The emphasis is on analytic transparency and threshold behavior rather than realistic estimation of specific fields' trajectories.

\subsection{The Field Lifetime Theorem}

As a field matures, its prior $\pi(t)$ tends to decline: genuine relationships are discovered and removed from the pool of open questions, while competitive pressure expands speculative hypotheses. Under the dynamics $\pi(t) = \pi_0 e^{-(\gamma+\delta)t}$:

\begin{theorem}[Field Lifetime]
\label{thm:field_lifetime}
Assume $\gamma + \delta > 0$, $\pi_0 > \pcrit(\tau, \alpha, \beta)$, and $\alpha$ and $\beta$ are held fixed. The field crosses $\pcrit$ at time
\begin{equation}
T^*(\tau) = \frac{1}{\gamma+\delta}\ln\!\left(\frac{\pi_0}{\pcrit(\tau, \alpha, \beta)}\right).
\end{equation}
For $t > T^*$, the target PPV is structurally unattainable.
\end{theorem}

\begin{proof}
Set $\pi_0 e^{-(\gamma+\delta)t} = \pcrit$ and solve.
\end{proof}

\begin{example}[Effect of Threshold Tightening on Field Lifetime]
A field with $\pi_0 = 0.70$, $\gamma+\delta = 0.05$, power $= 0.80$, target $\tau = 0.95$. At $\alpha = 0.05$: $\pcrit = 0.543$, $T^* \approx 5$ years. At $\alpha = 0.005$: $\pcrit = 0.106$, $T^* \approx 38$ years. Reducing $\alpha$ tenfold extends the productive lifetime sevenfold.
\end{example}

\subsection{Generational Dynamics and the Degenerative Programme Criterion}

When follow-up research builds on published findings, false positives in one generation degrade the priors available to the next.

\begin{theorem}[Generational Dynamics]
\label{thm:generational_decay}
Let $\PPV_k$ denote the PPV of generation $k$, with follow-up hypotheses conditional on a parent finding being true with probability $\pi_c \in (0,1)$. The effective prior for generation $k+1$ is $\pi_{k+1} = \PPV_k \cdot \pi_c$, and the PPV evolves as
\begin{equation}
\label{eq:gen_decay}
\PPV_{k+1} = \frac{\PPV_k \pi_c \Lambda}{\PPV_k \pi_c \Lambda + (1 - \PPV_k \pi_c)}.
\end{equation}
\begin{enumerate}[label=(\alph*)]
\item \textbf{Collapse} ($\pi_c\Lambda \leq 1$): $\PPV_k \to 0$ monotonically. In the boundary case $\Lambda = 1$, $\PPV_{k+1} = \pi_c \PPV_k$, so $\PPV_k = \pi_c^k \PPV_0 \to 0$ for $\pi_c < 1$.
\item \textbf{Recovery} ($\pi_c\Lambda > 1$): if $\PPV_0 > 0$, then $\PPV_k \to (\pi_c\Lambda - 1)/[\pi_c(\Lambda-1)]$.
\end{enumerate}
\end{theorem}

\begin{proof}
See Appendix~\ref{app:generational}.
\end{proof}

Algebraically, $\pi_c\Lambda \leq 1$ is identical to the Majority-False Threshold applied to follow-up research, providing a quantitative bridge to \citeauthor{lakatos1978}'s (\citeyear{lakatos1978}) concept of degenerative research programmes. Call $\pi_c\Lambda$ the programme's \textbf{progress ratio}.

\begin{corollary}[Degenerative Programme Criterion]
\label{cor:degenerative}
A research programme building on published findings is \emph{provably degenerative} (generating successive generations of decreasing reliability converging to noise) if and only if its progress ratio satisfies $\pi_c \Lambda \leq 1$.
\end{corollary}

A degenerative programme has placed itself in the majority-false regime for its own successors. When the progress ratio falls at or below~1, recovery requires either raising $\pi_c$ (restricting follow-up to higher-quality candidates) or raising $\Lambda$ (stricter thresholds or replication pipelines). Without such changes, each generation inherits a weaker prior and the literature drifts toward noise.

\begin{example}[Collapse and Recovery]
\textit{Collapse.} Candidate gene research with $\Lambda = 7$, $\pi_c = 0.10$: $\pi_c\Lambda = 0.7 < 1$.
\begin{center}
\begin{tabular}{cccc}
\toprule
Gen.\ $k$ & $\pi_k$ & $\PPV_k$ & False pos.\ per true pos. \\
\midrule
0 & 0.020 & 0.125 & 7.0 \\
1 & 0.013 & 0.081 & 11.3 \\
2 & 0.008 & 0.054 & 17.4 \\
3 & 0.005 & 0.037 & 26.2 \\
\bottomrule
\end{tabular}
\end{center}
\textit{Recovery.} GWAS with $\Lambda = 1.6\times10^7$, $\pi_c = 0.50$: $\pi_c\Lambda \gg 1$. PPV recovers rapidly.
\end{example}

\subsection{Self-Reinforcing Degeneration}
\label{subsec:self_reinforcing}

Once established, this feedback compounds. False positives spawn speculative follow-ups, lowering $\pi$ for the next generation; depressed PPV generates still more false positives, which lower $\pi$ further. Each link in this chain is an instance of the Certainty Bound applied to a degraded prior. The mechanism does not require increasing misconduct; it requires only that follow-up research treats published findings as informative about where to look next.

Fields in this regime cannot self-correct through incremental improvement alone. Larger samples and more careful statistical practice cannot arrest the decline unless they raise $\pi_c$ or materially increase effective leverage $\Lambda$. In the stylized dynamics here (fixed $\Lambda$ and fixed $\pi_c$), these interventions leave the structural parameters unchanged. Escape therefore requires an exogenous change in one of these parameters.

The candidate gene literature did not recover through gradual refinement; it required the discontinuous adoption of GWAS \citep{border2019}, which raised $\Lambda$ by six orders of magnitude. Particle physics's $5\sigma$ threshold \citep{cowan2011,atlas2012,cms2012} and the adoption of registered reports \citep{chambers2013} represent analogous exogenous interventions. The degenerative programme criterion identifies not only which fields are in difficulty but why recovery requires exogenous changes to the operating parameters rather than incremental refinement within the existing paradigm.

\begin{remark}[Scope of the Dynamics Model]
\label{rem:exogenous}
The self-reinforcing degeneration model is stylized: it assumes fixed $\Lambda$ and a deterministic transmission rule $\pi_{k+1} = \PPV_k \cdot \pi_c$. Real fields exhibit stochastic variation, partial updating, and heterogeneous subprogrammes. The model's value lies in identifying the qualitative threshold ($\pi_c\Lambda \leq 1$) below which the direction of drift is structurally determined, not in quantitative prediction of specific field trajectories.
\end{remark}

%==========================================================
% SECTION 8: RELIABILITY LANDSCAPE
%==========================================================
\section{A Reliability Landscape}
\label{sec:landscape}

\subsection{The Reliability Landscape and Kuhnian Transitions}

The reliability landscape is defined by two structural parameters: experimental leverage $\Lambda$ on one axis and prior probability $\pi$ on the other. The feasibility boundary for target $\tau$ is the curve $\pi\Lambda = \tau(1-\pi)/(1-\tau)$, equivalently $\Psi = 1$. Fields above this curve operate in the feasible regime; fields below, in the infeasible regime.

Normal science, in this landscape, operates at a fixed position. Incremental improvements (slightly larger samples, marginal improvements in measurement) produce continuous, small movements. A fundamental transition in methodology, however, produces a discontinuous rightward jump in $\Lambda$ that can cross the feasibility boundary in a single step. The GWAS transition jumped six orders of magnitude in $\Lambda$; particle physics's adoption of the $5\sigma$ criterion \citep{cowan2011,atlas2012,cms2012} provides an analogous illustration of threshold tightening as an institutional evidential standard. Neither was possible through incremental improvement; both required replacing the standard of evidence itself.

Movement across the feasibility boundary captures a quantitative signature of what \citeauthor{kuhn1962} (\citeyear{kuhn1962}) called a paradigm shift. The direction of explanation matters: boundary crossings do not \emph{define} paradigm shifts, but methodological revolutions often produce them. What Kuhn described qualitatively as revolution has, in this framework, a measurable signature: a discontinuous increase in leverage that crosses the $\Psi=1$ boundary.

\begin{figure}[htbp]
\centering
\begin{tikzpicture}
\begin{axis}[
    width=13cm, height=10cm,
    xmode=log,
    ymode=log,
    xlabel={Experimental Leverage $\Lambda = (1-\beta)/\alpha$},
    ylabel={Prior Probability $\pi$},
    xmin=1, xmax=1e8,
    ymin=1e-6, ymax=1.0,
    legend style={at={(0.03,0.03)}, anchor=south west, font=\small},
    grid=major,
    grid style={line width=0.3pt, gray!40},
    clip=false,
]
\addplot[thick, dashed, black, domain=1:1e8, samples=200]
    {19/(x + 19)};
\addlegendentry{Feasibility boundary ($\tau=0.95$, $\Psi=1$)}
\addplot[thick, dotted, black, domain=1:1e8, samples=200]
    {1/(x + 1)};
\addlegendentry{Majority-false boundary ($\PPV=0.50$)}
\addplot[gray, domain=1:1e8, samples=100]
    {4/(x + 4)};
\addlegendentry{$\PPV=0.80$ contour}

\addplot[only marks, mark=*, mark size=3pt, blue] coordinates {(7, 0.10)};
\addplot[only marks, mark=*, mark size=3pt, blue] coordinates {(10, 0.02)};
\addplot[only marks, mark=*, mark size=3pt, blue] coordinates {(12, 0.08)};
\addplot[only marks, mark=*, mark size=3pt, red] coordinates {(16, 0.30)};
\addplot[only marks, mark=*, mark size=3pt, red] coordinates {(18, 0.25)};
\addplot[only marks, mark=*, mark size=3pt, green!60!black] coordinates {(1.6e7, 1e-5)};
\addplot[only marks, mark=*, mark size=3pt, green!60!black] coordinates {(3.3e6, 0.90)};

\draw[blue, thin] (axis cs:7, 0.10) -- (axis cs:500, 0.10);
\node[blue, font=\small, anchor=west] at (axis cs:500, 0.10) {Pre-reform psych};
\draw[blue, thin] (axis cs:12, 0.08) -- (axis cs:500, 0.03);
\node[blue, font=\small, anchor=west] at (axis cs:500, 0.03) {Nutritional epi};
\draw[blue, thin] (axis cs:10, 0.02) -- (axis cs:500, 0.008);
\node[blue, font=\small, anchor=west] at (axis cs:500, 0.008) {Candidate genes};
\draw[red, thin] (axis cs:16, 0.30) -- (axis cs:500, 0.50);
\node[red, font=\small, anchor=west] at (axis cs:500, 0.50) {Well-powered RCT};
\draw[red, thin] (axis cs:18, 0.25) -- (axis cs:500, 0.25);
\node[red, font=\small, anchor=west] at (axis cs:500, 0.25) {Pre-reg psych};
\node[green!60!black, font=\small, anchor=north west] at (axis cs:1.8e7, 8e-6) {GWAS};
\node[green!60!black, font=\small, anchor=south west] at (axis cs:3.8e6, 0.92) {Particle physics};
\end{axis}
\end{tikzpicture}
\caption{Reliability landscape for scientific fields. Points represent illustrative calibrations, not precise measurements; their purpose is regime visualization. The dashed line is the feasibility boundary for $\tau = 0.95$ ($\Psi = 1$): fields above and to the right are feasible. The dotted line is the majority-false boundary. Blue: majority-false exemplars. Red: improved but still infeasible. Green: feasible.}
\label{fig:landscape}
\end{figure}

\subsection{Field Calibration}

As a field matures, $\pi$ declines and the field drifts downward on the vertical axis. Without compensating threshold tightening (a rightward shift), it crosses the feasibility boundary. Table~\ref{tab:fields} maps the landscape across representative fields.

The purpose of Table~\ref{tab:fields} is regime visualization under plausible parameterizations, not precise field-level estimation.

\begin{table}[htbp]
\centering
\caption{Illustrative regime-mapping scenarios for representative fields at target $\tau = 0.95$.}\label{tab:fields}
\small
\setlength{\tabcolsep}{0pt}
\begin{tabular*}{\textwidth}{@{\extracolsep{\fill}} lcccccccc}
\toprule
\textbf{Field} & $\alpha$ & \textbf{Power} & $\pi$ & $\Lambda$ & $\Psi$ & \textbf{PPV} & \textbf{Ceiling} & \textbf{Regime} \\
\midrule
Candidate genes & 0.05 & 0.50 & 0.02 & 10 & 93 & 17\% & 29\% & Maj.-false \\
Pre-reform psych & 0.05 & 0.35 & 0.10 & 7 & 24 & 44\% & 69\% & Infeasible \\
Nutritional epi & 0.05 & 0.60 & 0.08 & 12 & 18 & 51\% & 63\% & Infeasible \\
Well-powered RCT & 0.05 & 0.80 & 0.30 & 16 & 2.8 & 87\% & 90\% & Infeasible \\
Pre-reg psych & 0.05 & 0.90 & 0.25 & 18 & 3.2 & 86\% & 87\% & Infeasible \\
GWAS & $5\!\times\!10^{-8}$ & 0.80 & $10^{-5}$ & $1.6\!\times\!10^7$ & 0.12 & 99\% & $\!\sim\!100\%$ & Feasible \\
Particle physics & $3\!\times\!10^{-7}$ & 0.9999 & 0.90 & $3.3\!\times\!10^6$ & $6\!\times\!10^{-7}$ & $\!\sim\!100\%$ & $\!\sim\!100\%$ & Feasible \\
\bottomrule
\end{tabular*}
\smallskip
{\footnotesize\textit{Note.} All $\pi$ values are illustrative calibrations, not precise estimates. For pre-reform psychology, $\pi = 0.10$ is a calibration parameter consistent with scenarios in \citet{ioannidis2005} and the retrodiction in Section~\ref{sec:replication}. Candidate gene calibration is chosen to be consistent with the well-documented failure of the literature \citep{border2019}; this row is among the most robust because the literature's failure is well documented. Nutritional epidemiology values draw on \citet{ioannidis2018} and are illustrative; the infeasibility classification is robust across $\pi \in [0.03, 0.20]$. The RCT row is the most sensitive: it crosses into feasible only if $\pi \gtrsim 0.54$ at the stated parameters. Pre-registered psychology ($\pi \approx 0.25$) reflects the conjecture that pre-registration accompanies more theory-driven research with higher base rates; this is the most uncertain estimate. GWAS parameters reflect Bonferroni correction for $\sim 10^6$ variants \citep{peer2008}. Particle physics uses an illustrative $5\sigma$-level threshold (one-sided $\alpha \approx 3\times 10^{-7}$) \citep{cowan2011}. Sensitivity to $\pi$ is explored in Table~\ref{tab:sensitivity}.}
\end{table}

Even pre-registered psychology remains in the infeasible regime at $\Psi \approx 3.2$, suggesting that pre-registration alone is insufficient without either stricter thresholds or explicit replication requirements.

%==========================================================
% SECTION 9: STRUCTURAL REQUIREMENTS
%==========================================================
\section{Structural Requirements for Reliable Research}
\label{sec:design}

For fields aiming at a declared reliability target $\tau$ under binary significance-based publication architectures, four necessary conditions follow for operating in the feasible regime. Within this framework, these are design requirements implied by the mathematics, not discretionary guidelines.

\textbf{1. Threshold calibration.} A field's significance threshold must satisfy
\begin{equation}
\label{eq:alpha_req}
\alpha \leq \alpha_{\max}(\pi, \beta, \tau) = \frac{(1-\beta)(1-\tau)\pi}{\tau(1-\pi)}.
\end{equation}
For $\pi = 0.10$, $\tau = 0.95$, power 0.80: $\alpha_{\max} = 0.0047$. The \citet{benjamin2018} proposal ($p < 0.005$) is consistent with this requirement.

\textbf{2. Enforcement of the nominal level.} Pre-registration and registered reports enforce $\aeff \approx \alpha$ (Corollary~\ref{cor:prereg}). The nominal $\alpha$ is only binding if the effective $\aeff$ is controlled.

\textbf{3. Replication as the unit of evidence.} At conventional $\alpha = 0.05$ and target $\tau = 0.95$, the Fixed-$\alpha$ Ceiling falls below the target unless $\pi \gtrsim 0.49$. For the many fields operating well below this threshold, the standard must therefore be a replication pipeline: $k$ independent studies yield leverage $\Lambda^k$ (Theorem~\ref{thm:rep_leverage}). The scientific finding is the outcome of the pipeline.

\textbf{4. Valid identification for causal claims.} No threshold or sample size overcomes Observational Collapse (Theorem~\ref{thm:obs_collapse}). Observational studies can produce useful evidence for associations; it is the interpretation of significant associations as causal effects that requires valid identification.

Not all reforms are equally effective. Pre-registration controls $\aeff$ but does not change $\pi$ or nominal $\Lambda$. Threshold tightening increases $\Lambda$ directly. Replication multiplies $\Lambda$ geometrically. Increasing sample size under fixed $\alpha$ is bounded by the ceiling. Exhortations to ``better practice'' that do not change the operating parameters $(\pi, \alpha, 1-\beta)$ cannot move $\Psi$ and therefore cannot by themselves move a field across the feasibility boundary.

\subsection{A Worked Example: Preclinical Alzheimer's Research}

To illustrate these requirements concretely, consider the preclinical Alzheimer's literature as a stylized case. Suppose the prior probability of a preclinical target hypothesis being correct is $\pi \approx 0.05$, which is plausibly optimistic given the high attrition rate in CNS drug development (see \citealp{cummings2014}). Typical preclinical studies often operate at $\alpha = 0.05$, and low to moderate power is common in adjacent preclinical and neuroscience literatures \citep{button2013}; we use power $\approx 0.50$ here as an illustrative value.

At these parameters: $\Lambda = 10$, $\PPV = 0.34$, $\Psi = 36$, and the Fixed-$\alpha$ Ceiling is 51\%. The field is in the majority-false regime. More than half of its significant preclinical findings are expected to be false positives, not because of incompetent researchers, but because of the operating parameters.

What would repair look like? Tightening to $\alpha = 0.005$ raises $\Lambda$ to 100 and PPV to 84\%, but $\Psi$ remains above~1. A pipeline of $k = 2$ independent replications at $\alpha = 0.005$ yields $\Lambda^{(2)} = 10{,}000$ and PPV $> 99\%$. Authors could compute $\Psi$ at the design stage; journals could require its disclosure alongside standard power analyses; and $\Psi > 1$ without a replication plan could be treated as a design limitation requiring revision.

More broadly, a minimal evidential status report would include: the target reliability $\tau$; the assumed prior range; $\alpha$, power, and $\Psi$; planned replication depth $k$; and identification status for causal claims.

%==========================================================
% SECTION 10: DISCUSSION
%==========================================================
\section{Discussion}
\label{sec:discussion}

\subsection{Relation to Existing Work}

\citet{ioannidis2005} argued that many published findings are false under realistic assumptions. The present paper extends that line of argument in three directions: from diagnosis to structural infeasibility, from static snapshots to field dynamics, and from analytical critique to an empirical bridge connecting the framework to observed replication rates. Table~\ref{tab:comparison} summarizes the distinctions.

\begin{table}[htbp]
\centering
\caption{Comparison with prior approaches.}
\label{tab:comparison}
\small
\begin{tabularx}{\textwidth}{@{} >{\raggedright\arraybackslash}p{2.5cm} >{\raggedright\arraybackslash}X >{\raggedright\arraybackslash}X >{\raggedright\arraybackslash}X >{\raggedright\arraybackslash}X @{}}
\toprule
 & \textbf{Ioannidis (2005)} & \textbf{Selection models} & \textbf{Diagnostic tools} & \textbf{This paper} \\
\midrule
\textbf{Core question} & How often false? & True effect size? & Real signal? & Reliability achievable? \\
\addlinespace
\textbf{PPV ceiling} & Implicit & Not targeted & Not targeted & Explicit \\
\addlinespace
\textbf{Collapse} & None & Publication filter & None & Two collapse routes proved \\
\addlinespace
\textbf{Dynamics} & None & None & None & Degenerative criterion \\
\addlinespace
\textbf{Philosophy} & None & None & None & Popper, Kuhn, Lakatos \\
\bottomrule
\end{tabularx}
\end{table}

Prior work on publication bias spans selection models \citep{hedges1984, andrews2019} and behavioral models \citep{simmons2011, gelman2014}. Selection models estimate bias-corrected effect sizes; the Certainty Bound asks what PPV is structurally attainable given the selection environment. The two are compatible: one could use Andrews-Kasy-corrected estimates to recalibrate $\pi$ and power, then apply the Certainty Bound to the corrected parameters.

Three diagnostic tools deserve specific comparison. \emph{P-curve} \citep{simonsohn2014} tests whether a set of significant $p$-values has evidential value; \emph{z-curve} \citep{bartos2022} estimates expected replication and discovery rates. Both assess the evidential content of an existing literature. The Certainty Bound asks a complementary question: could this literature contain reliable signal given its design parameters? A literature might pass a $p$-curve test (the effects are real) while remaining in the infeasible regime (PPV too low for the reliability target).

\subsection{Connections to the Philosophy of Science}

The philosophical connections developed below are interpretive bridges, not claims of exact equivalence between the formalism and the historical philosophical accounts. The aim is not to reduce Popper, Kuhn, or Lakatos to a single metric, but to show that the framework supplies a quantitative structure for reliability constraints that these traditions describe qualitatively.

\subsubsection{Popper, Mayo, and the Severity Criterion}

\citet{popper1959} held that corroborated hypotheses (those surviving serious attempts at falsification) deserve provisional acceptance. The Certainty Bound sharpens this inference by specifying when such acceptance is epistemically warranted: when the testing procedure has sufficient leverage relative to the prior, as summarized by the infeasibility ratio $\Psi$.

The connection to \citeauthor{mayo1996}'s (\citeyear{mayo1996}) severity criterion is particularly close. Mayo's error-statistical framework holds that a hypothesis passes a severe test only when the procedure had high probability of detecting error if present: low $\alpha$ (the procedure rarely signals when nothing is there) and high $1-\beta$ (the procedure reliably detects genuine effects). Leverage $\Lambda = (1-\beta)/\alpha$ is a quantitative expression of this severity-like discrimination; it combines both components into a single ratio measuring how much more probable a significant result is under truth than under the null.

The qualification ``severity-like'' is deliberate: Mayo's full severity concept additionally involves the specificity of the error probed, a dimension not captured by the scalar $\Lambda$ alone. Popper required that corroborating tests be severe but lacked a formal apparatus for specifying \emph{how} severe; Mayo provided the conceptual framework; the Certainty Bound supplies a threshold. A test meets this severity-like standard for target reliability $\tau$ if and only if $\Psi < 1$. Below that threshold, a hypothesis that achieves significance has not been severely tested in this sense: the procedure was not capable of delivering the claimed reliability given the prior.

In the majority-false regime, even a hypothesis achieving significance at $\alpha = 0.05$ with reasonable power is more likely to be a false positive than a genuine effect. Falsification remains logically valid, but its epistemic yield is constrained by leverage.

\subsubsection{Kuhn and Paradigm Shifts}

In this landscape (Figure~\ref{fig:landscape}), normal science operates at a fixed position; paradigm shifts are discontinuous jumps that cross the $\Psi = 1$ boundary. A boundary crossing is not the \emph{definition} of a paradigm shift; it is a measurable quantitative signature often produced by methodological revolutions. The GWAS transition and the $5\sigma$ criterion in particle physics are the canonical examples: fields that crossed not by improving individual studies but by changing what counts as evidence. For the broader confirmation-theoretic context (including Bayesian debates), see \citet{earman1992}; for a wider philosophy-of-science overview, see \citet{gillies1993}.

\subsubsection{Lakatos and Degenerative Programmes}

The Degenerative Programme Criterion (Corollary~\ref{cor:degenerative}) gives \citeauthor{lakatos1978}'s (\citeyear{lakatos1978}) distinction between progressive and degenerative programmes a quantitative threshold: $\pi_c\Lambda \leq 1$. The criterion captures the reliability dimension of Lakatosian degeneracy, the empirical signature that research output converges to noise, rather than the full concept, which additionally involves theoretical stagnation and post-hoc adjustment of auxiliary hypotheses. When $\pi_c\Lambda \leq 1$, recovery requires exogenous changes to the operating parameters rather than incremental effort within the paradigm. The candidate gene-to-GWAS transition is both a paradigm shift in Kuhn's sense and an escape from degeneration in Lakatos's; the Certainty Bound supplies the threshold in each case.

\subsection{The Informativeness Paradox}

One implication of the analysis runs against usual intuition: in the majority-false regime, non-significant results are \emph{more informative} than significant ones. 

\begin{proposition}[Null Result Informativeness]
\label{prop:npv}
If the test has discrimination ($1 - \beta > \alpha$) and $\PPV < 1/2$, then $\NPV > 1/2 > \PPV$: non-significant results are more reliable indicators of the true state than significant ones.
\end{proposition}

\begin{proof}
See Appendix~\ref{app:npv}.
\end{proof}

For pre-reform psychology ($\alpha = 0.05$, $\beta = 0.65$, $\pi = 0.10$), the model gives $\NPV = 0.93$ and $\PPV = 0.44$. Under this parameterization, a null result is 93\% likely to be correct, whereas a significant result is only 44\% likely to be correct. Journals filtering for positive findings are therefore preferentially selecting the less reliable outcome class. In the majority-false regime, the file drawer does not merely hide ``nulls''; it hides reliability.

Two further implications follow. First, the publication filter produces a literature that actively misleads: false positives outnumber true positives among published significant findings. Second, the misinformation rate has a floor that cannot be reduced by increasing sample size under fixed $\alpha$:
\begin{equation}
\label{eq:misinfo_floor}
1 - \PPV \;\geq\; 1 - \PPV^{\mathrm{ceil}} = \frac{\alpha(1-\pi)}{\pi + \alpha(1-\pi)}.
\end{equation}
In fields below the majority-false threshold, pre-registering and publishing null results would shift the literature toward its more informative segment. The reliability of the published literature would improve not because more studies are run, but because the publication filter would no longer suppress the more reliable outcome class in this regime.

\subsection{Meta-Analysis}

When contributing studies each have effective leverage close to~1 and share the same identification failure, a meta-analysis can inherit the collapse: pooling then accumulates precision around a biased estimand rather than restoring valid identification. In the collapse limit, each low-leverage input contributes little evidential discrimination, so aggregation alone does not restore identification. Meta-analysis remains a powerful tool when included studies maintain adequate leverage and valid identification; the collapse-inheritance point applies specifically to low-leverage inputs with shared failure modes. Evidence hierarchies placing meta-analysis at the apex should therefore be conditional on the leverage and identification quality of included studies, not merely their number.

\subsection{Limitations and Open Problems}

Prior probabilities are not directly observable, so calibrations rely on meta-science. The framework does not, however, require precise $\pi$: it specifies what $\pi$ must be for a given testing configuration to achieve the claimed reliability. The independence assumption is a modeling benchmark (Remark~\ref{rem:independence}); dependence across studies can alter finite-sample leverage in either direction.

The analysis addresses binary significance claims, a deliberate coarsening of the full evidence available in any study. Continuous estimation and Bayesian methods can extract more information; the present results characterize the ceiling imposed by binary publication architectures. The framework does not claim that all fields have a single $\pi$, that all studies are independent, or that continuous evidence is bounded in this way. It provides feasibility constraints for binary significance architectures under stated operating parameters. The prior $\pi$ refers throughout to the probability that a hypothesis \emph{selected for testing} is true, a property of the field's hypothesis-generation process, not to the fraction of all conceivable hypotheses that are true.

The architectural and behavioral accounts are analytically distinct but empirically coupled: incentive structures lower $\pi$, while behavioral reforms such as registered reports \citep{chambers2013} enforce $\aeff \approx \alpha$ and may raise $\pi$ indirectly by reshaping which hypotheses are pursued. The architectural framework characterizes the constraints; behavioral reform is one mechanism for moving within them.

Three open problems follow naturally: empirical estimation of $\pi$ via bridge inversion across fields, asymmetric error cost structures, and PPV properties of sequential designs under publication pressure.

%==========================================================
% SECTION 11: CONCLUSION
%==========================================================
\section{Conclusion}
\label{sec:conclusion}

Within binary significance-based publication architectures, the Certainty Bound implies that a substantial component of the replication crisis is structural. Parameter ranges for pre-reform psychology, documented before the Open Science Collaboration's project, imply a replication rate of approximately 36\%, consistent with the observed figure. On this account, improving reliability requires structural reform: thresholds calibrated to field priors, nominal significance levels enforced through pre-registration or registered reports, replication pipelines institutionalized as the unit of evidence, and valid identification secured for causal claims in observational research.

The framework's connections to the philosophy of science are substantive rather than ornamental. It supplies quantitative conditions under which Popperian falsification is epistemically informative, gives Kuhnian methodological revolutions a measurable signature in the reliability landscape, and provides Lakatosian degenerative programmes with a reliability threshold.

The practical program is straightforward: compute $\Psi$ at the design stage, disclose it alongside standard power analyses, require replication pipeline standards where $\Psi > 1$, and distinguish associational from causal claims when identification is not secured. Within this framework, these are not optional refinements; they are design requirements implied by the mathematics of binary significance-based evidence.

Within a binary significance-based architecture, a field that ignores these requirements cannot achieve the reliability it claims, regardless of its researchers' diligence. That is not a moral indictment. It is a design diagnosis.

\newpage

\section*{Data, Code, and Materials Availability}

This manuscript is primarily theoretical and does not report analyses of an original dataset. The results, figures, and tables are derived from the formulas and parameter settings stated in the manuscript. An interactive web tool implementing the $\Psi$ diagnostic, replication pipeline calculator, and reliability landscape visualizer is available at \url{https://mpollanen.github.io/certainty-bound-tool/} (archived at \url{https://osf.io/c5wun}). Source code for the interactive tool is available via the archived project under a CC-BY 4.0 license. The tool runs entirely client-side; no data are transmitted.

\section*{Conflict of Interest and Funding}

The author declares no conflict of interest. The author acknowledges the support of the Natural Sciences and Engineering Research Council of Canada (NSERC), funding reference number RGPIN-2019-04085.

\section*{Author Contributions}

Marco Pollanen: Conceptualization, Formal analysis, Investigation, Methodology, Writing -- original draft, Writing -- review \& editing.

%==========================================================
% APPENDIX
%==========================================================
\newpage
\appendix
\section*{Appendix: Proofs of Secondary Results}

\subsection{Proof of Proposition~\ref{prop:jensen} (Prior Heterogeneity Penalty)}
\label{app:jensen}

Write $\PPV(\pi) = \pi\Lambda/[\pi(\Lambda-1)+1]$. Differentiating twice: $\PPV''(\pi) = -2\Lambda(\Lambda-1)/[\pi(\Lambda-1)+1]^3 < 0$ for $\Lambda > 1$, establishing strict concavity. Jensen's inequality gives $\E[\PPV(\Pi)] < \PPV(\bar\pi)$ (strict when $\Pi$ is non-degenerate and $\Lambda > 1$). A second-order Taylor expansion around $\bar\pi$ yields the approximation $\PPV(\bar\pi) - \E[\PPV(\Pi)] \approx \Lambda(\Lambda-1)\sigma^2_\pi/[\bar\pi(\Lambda-1)+1]^3$. \qed

\subsection{Proof of Theorem~\ref{thm:cost} (Cost of Discovery)}
\label{app:cost}

Under the i.i.d.\ assumption, each study yields a true positive with probability $\pi(1-\beta)$ and a false positive with probability $(1-\pi)\alpha$. Over $N$ independent studies, the expected counts are $N\pi(1-\beta)$ and $N(1-\pi)\alpha$, respectively, so their ratio is $(1-\pi)\alpha/[\pi(1-\beta)]$. This equals $(1-\PPV)/\PPV$ by direct substitution from the PPV formula. \qed

\subsection{Proof of Theorem~\ref{thm:obs_collapse} (Observational Collapse)}
\label{app:obs_collapse}

(a) Under $H_{\mathrm{causal}}=0$: $Z_n$ is asymptotically $N(\sqrt{n}\,b_n/\sigma, 1)$ by assumption. Since $\sqrt{n}\,b_n/\sigma \to \infty$, $\aeff(n) = \Prob(Z_n > z_\alpha \mid H_{\mathrm{causal}}=0) \to 1$.

(b) Under $H_{\mathrm{causal}}=1$: $1-\beta_n \to 1$ by consistency, so $\Leff = (1-\beta_n)/\aeff(n) \to 1$.

(c) By the Certainty Bound: $\PPV = \pi\Leff/[\pi\Leff + (1-\pi)] \to \pi$.

(d) For two-sided rejection ($|Z_n| > z_{\alpha/2}$): if $\sqrt{n}\,|b_n| \to \infty$, then $|Z_n| \to \infty$ in probability under $H_{\mathrm{causal}} = 0$, so $\aeff(n) \to 1$. \qed

\subsection{Proof of Lemma~\ref{lem:effective_errors} (Effective Error Rates)}
\label{app:spec_search}

Under $H_0$, the probability of significance on attempt $k$ is $\alpha \cdot s_0^{k-1}$ where $s_0 = (1-\alpha)(1-q)$. Summing over $k=1,\ldots,m$: $\aeff = \alpha(1-s_0^m)/(1-s_0)$. The analogous argument with $s_1 = \beta(1-q)$ gives the effective power formula. \qed

\subsection{Proof of Theorem~\ref{thm:ppv_adj} (Discrimination Loss)}
\label{app:disc_loss}

Substituting effective rates into the Certainty Bound:
\[
\Lambda_{\mathrm{eff}}^{\mathrm{pb}} = \frac{(1-\beta)\frac{1-s_1^m}{1-s_1}}{\alpha\frac{1-s_0^m}{1-s_0}} = \Lambda \cdot \underbrace{\frac{1-s_1^m}{1-s_0^m} \cdot \frac{1-s_0}{1-s_1}}_{D(q,m)}.
\]
Define $h(s) := \frac{1-s^m}{1-s} = \sum_{j=0}^{m-1} s^j$. Since $h$ is strictly increasing on $[0,1)$ and $s_1 < s_0$ (because $1-\beta > \alpha$ implies $\beta < 1-\alpha$, so $\beta(1-q) < (1-\alpha)(1-q)$), we have $h(s_1) < h(s_0)$, giving $D = h(s_1)/h(s_0) < 1$. \qed

\subsection{Proof of Proposition~\ref{prop:spec_collapse} (Saturation and Collapse)}
\label{app:saturation}

(a) With $q \in (0,1)$: $s_0, s_1 \in (0,1)$, so $s_0^m, s_1^m \to 0$. In the limit, $D(q,\infty) = (1-s_0)/(1-s_1)$ and $\Lambda_{\mathrm{sat}} = \Lambda \cdot D(q,\infty)$. This satisfies $1 < \Lambda_{\mathrm{sat}} < \Lambda$ because $q[(1-\beta)-\alpha] > 0$.

(b) With $q = 0$: $s_0 = 1-\alpha$, $s_1 = \beta$. As $m \to \infty$: $\aeff \to 1$ and $1-\beta_{\mathrm{eff}} \to 1$. Hence $\Lambda_{\mathrm{eff}}^{\mathrm{pb}} \to 1$ and $\PPV \to \pi$. \qed

\subsection{Proof of Theorem~\ref{thm:double_collapse} (Double Collapse)}
\label{app:double_collapse}

(a) Under the stated assumptions, $\aeff(n,m,q) \to 1$ and $1-\beta_{\mathrm{eff}}(n,m,q) \to 1$, hence $\Leff(n,m,q) = (1-\beta_{\mathrm{eff}}(n,m,q))/\aeff(n,m,q) \to 1$.

(b) Under specification search alone: by Proposition~\ref{prop:spec_collapse}(b), both rates $\to 1$, so $\Leff \to 1$. \qed

\subsection{Proof of Theorem~\ref{thm:adaptive_escape} (Adaptive Escape)}
\label{app:adaptive_escape}

With $\alpha_n = 1 - \Phi(c\sqrt{n})$: since $c\sqrt{n} \to \infty$, $\alpha_n \to 0$. Under $H_1$, $Z_n \sim N(\theta_1\sqrt{n}/\sigma, 1)$, so $1-\beta_n = \Phi((\theta_1/\sigma - c)\sqrt{n}) \to 1$ since $c < \theta_1/\sigma$. Thus $\Lambda_n \to \infty$. The required sample size for $\Lambda_n \geq \Lreq$ is approximately $n \approx (2/c^2)\ln\Lreq$ since $\alpha_n$ decays exponentially. \qed

\subsection{Proof of Theorem~\ref{thm:generational_decay} (Generational Dynamics)}
\label{app:generational}

Substituting $\pi_{k+1} = \PPV_k\pi_c$ into the Certainty Bound gives $\PPV_{k+1} = f(\PPV_k)$ where $f(x) = x\pi_c\Lambda/(x\pi_c\Lambda + 1 - x\pi_c)$.

\textit{Fixed points.} Setting $f(x)=x$: either $x=0$ or $x^* = (\pi_c\Lambda-1)/[\pi_c(\Lambda-1)]$, with the positive fixed point existing iff $\pi_c\Lambda > 1$.

\textit{Collapse ($\pi_c\Lambda \leq 1$).} The numerator factor $g(x) = \pi_c\Lambda - 1 - x\pi_c(\Lambda-1)$ satisfies $g(x) \leq 0$ for all $x \in (0,1]$, whether $\Lambda \geq 1$ or $\Lambda < 1$. Hence $f(x) \leq x$, the sequence is non-increasing, bounded below by~0, and the only fixed point is $x=0$. In the boundary case $\Lambda = 1$, $f(x) = \pi_c x$, so $\PPV_k = \pi_c^k \PPV_0 \to 0$ for $\pi_c < 1$.

\textit{Recovery ($\pi_c\Lambda > 1$).} For $x \in (0, x^*)$: $f(x) > x$ (increasing toward $x^*$). For $x > x^*$: $f(x) < x$ (decreasing toward $x^*$). By monotone convergence, $\PPV_k \to x^*$. \qed

\subsection{Proof of Proposition~\ref{prop:npv} (Null Result Informativeness)}
\label{app:npv}

$\NPV = (1-\pi)(1-\alpha)/[(1-\pi)(1-\alpha)+\pi\beta]$. We have $\PPV < 1/2$ iff $\pi < \alpha/[(1-\beta)+\alpha]$ and $\NPV > 1/2$ iff $\pi < (1-\alpha)/[(1-\alpha)+\beta]$. The first threshold is below the second iff $\alpha[(1-\alpha)+\beta] < (1-\alpha)[(1-\beta)+\alpha]$, which simplifies to $1-\beta > \alpha$ (the discrimination condition). Hence whenever $\PPV < 1/2$, we also have $\NPV > 1/2$. \qed

\end{document}